\documentclass[letterpaper, 10pt,conference]{ieeeconf}      

\IEEEoverridecommandlockouts                              

\usepackage{amsmath,amssymb,amsfonts,mathabx}
\usepackage{cite}
\usepackage[colorlinks = true, linkcolor = blue, citecolor =
blue]{hyperref}

\let\labelindent\relax
\usepackage{algorithm,algorithmic}
\usepackage{enumitem}
\usepackage[font = small]{subcaption}

\usepackage[english]{babel}
\usepackage{graphicx}
\graphicspath{{./},{./}}

\newcommand{\Mc}[1]{\mathcal{#1}}
\newcommand{\agt}{\mathcal{V}} 
\newcommand{\agte}{\mathcal{V}_{\mathrm{E}}} 
\newcommand{\agtde}{\mathcal{V}_{\mathrm{DE}}}

\newcommand{\tp}{\mathcal{T}} 
\newcommand{\zij}{z_{i}^{j}} 
\newcommand{\zkj}{z_{k}^{j}}
\newcommand{\zi}{z_{i}}

\newcommand{\znoti}{\mathbf{z}_{-i}}

\newcommand{\aik}{a_{ik}}

\newcommand{\aii}{a_{ii}}
\newcommand{\aij}{a_{ij}}

\newcommand{\Bi}{B_{i}}

\newcommand{\boldzi}{\mathbf{z}_{i}} 
\newcommand{\boldzk}{\mathbf{z}_{k}}
\newcommand{\boldzj}{\mathbf{z}_{j}}
\newcommand{\boldzl}{\mathbf{z}_{l}}
\newcommand{\boldz}{\mathbf{z}}
\newcommand{\boldx}{\mathbf{x}}
\newcommand{\boldxi}{\mathbf{x}_{i}}

\newcommand{\boldy}{\mathbf{y}}

\newcommand{\bolds}{\mathbf{s}}

\newcommand{\boldznoti}{\mathbf{z}_{-i}}
\newcommand{\boldznotj}{\mathbf{z}_{-j}}

\newcommand{\boldznotk}{\mathbf{z}_{-k}}
\newcommand{\boldg}{\mathbf{g}}
\newcommand{\boldp}{\mathbf{p}}
\newcommand{\boldv}{\mathbf{v}}
\newcommand{\boldq}{\mathbf{q}}
\newcommand{\boldpi}{\mathbf{p}_{i}}

\newcommand{\boldci}{\mathbf{c}_{i}}

\newcommand{\boldDi}{\mathbf{D}_{i}}

\newcommand{\boldfi}{\mathbf{f}_{i}}

\newcommand{\boldf}{\mathbf{f}}
\newcommand{\boldD}{\mathbf{D}}
\newcommand{\boldJ}{\mathbf{J}}

\newcommand{\tildeDi}{\widetilde{\mathbf{D}}_{i}}

\newcommand{\tildepi}{\widetilde{\mathbf{p}}_{i}}
\newcommand{\tildepk}{\widetilde{\mathbf{p}}_{k}}

\renewcommand{\tilde}[1]{\widetilde{#1}}

\newcommand{\identitym}{\mathbf{I}_{m}}
\newcommand{\identitymn}{\mathbf{I}_{mn}}

\newcommand{\support}[1]{\mathsf{supp}(#1)}
\newcommand{\innerproduct}{\langle \: \boldf(\boldz),\boldx-\boldz\ \rangle}

\newcommand{\innerproducti}{\langle \: \boldfi(\boldzi,\boldznoti),\boldxi-\boldzi\ \rangle}

\newcommand{\innerproductistar}{\langle \: \boldfi(\boldzi^*,\boldznoti^*),\boldxi-\boldzi^*\ \rangle}

\newcommand{\dij}{p_{i}^{j}}

\newcommand{\wi}{w_{i}}
\newcommand{\wij}{w_{i}^{j}}

\newcommand{\ci}{\mathbf{c}_{i}}

\newcommand{\laplacian}{\mathbf{L}}
\newcommand{\constraintseti}{\mathcal{K}_{i}}
\newcommand{\constraintsetnoti}{\mathcal{K}_{-i}}
\newcommand{\constraintsetl}{\mathcal{K}_{l}}
\newcommand{\constraintset}{\mathcal{K}}

\newcommand{\adjmat}{\mathbf{A}}
\newcommand{\Ni}{\mathcal{N}_{i}}
\newcommand{\Nie}{\mathcal{N}_{i}^{e}}
\newcommand{\Nif}{\mathcal{N}_{i}^{f}}
\newcommand{\agtsumi}{\sum_{i \in \agt}}
\newcommand{\agtsumk}{\sum_{k \in \agt}}

\newcommand{\tpsumj}{\sum_{j \in \tp}}

\newcommand{\enemysumi}{\sum_{k \in\mathcal{N}_{i}^{e}}}

\newcommand{\agtsumknoti}{\sum_{k \in \agt \setminus \left\{i\right\}}}

\newcommand{\pdsi}{\mathrm{PDS}(\boldfi,\constraintseti,\boldznoti)}
\newcommand{\pds}{\mathrm{PDS}(\boldf,\constraintset)}

\newcommand{\varineqfi}{\mathrm{VI}(\boldfi,\constraintseti,\boldznoti)}
\newcommand{\solvarineqfi}{\mathrm{SOL}(\boldfi,\constraintseti,\boldznoti)}

\newcommand{\solvarineqfistar}{\mathrm{SOL}(\boldfi,\constraintseti,\boldznoti^*)}

\newcommand{\varineqf}{\mathrm{VI}(\boldf,\constraintset)}
\newcommand{\solvarineqf}{\mathrm{SOL}(\boldf,\constraintset)}

\newcommand{\solvarineq}{\mathrm{SOL}(\mathbf{g},\Mc{X})}

\newcommand{\tki}{T_{\constraintseti}(\boldzi)}
\newcommand{\tk}{T_{\constraintset}(\boldz)}
\newcommand{\projtki}{\mathcal{P}_{\tki}}
\newcommand{\projtk}{\mathcal{P}_{\tk}}

\newcommand{\realn}{\mathbb{R}^{n}} 
\newcommand{\realm}{\mathbb{R}^{m}} 
\newcommand{\realmn}{\mathbb{R}^{mn}} 
\newcommand{\real}{\mathbb{R}}

\newcommand{\boldzero}{\mathbf{0}}

\newcommand{\edg}{\mathcal{L}}

\newcommand{\ldef}{:=}
\newcommand{\rdef}{=:}

\DeclareMathOperator*{\argmax}{argmax}
\DeclareMathOperator*{\argmin}{argmin}

\newcommand{\thmtitle}[1]{\mbox{}\textit{(\textbf{#1}.)}}

\newcommand{\remend}{\relax\ifmmode\else\unskip\hfill\fi\hbox{$\bullet$}}

\newcommand{\bulletsym}{\hbox{$\bullet$}}
\newcommand{\bulletend}{\relax\ifmmode\else\unskip\hfill\fi\bulletsym}
\newcommand{\squaresym}{\hbox{$\blacksquare$}}
\newcommand{\proofend}{\relax\ifmmode\else\unskip\hfill\fi\squaresym}
\newcommand{\trianglesym}{\hbox{$\blacktriangle$}}
\newcommand{\egend}{\relax\ifmmode\else\unskip\hfill\fi\trianglesym}

\def \eqpt{\mathcal{E}}

\def \grph{\mathsf{G}}

\def \ne{\mathcal{NE}}

\newcommand{\tth}{^{\text{th}}}

\renewenvironment{proof}{\textit{Proof:} }{\proofend}

\newtheorem{theorem}{\textbf{Theorem}}[section]
\newtheorem{corollary}[theorem]{\textbf{Corollary}}
\newtheorem{lemma}[theorem]{\textbf{Lemma}}

\newtheorem{remark}[theorem]{\textbf{Remark}}
\newtheorem{define}[theorem]{\textbf{Definition}}

\renewcommand{\baselinestretch}{0.985}
\title{\LARGE \bf Multi-Topic Projected Opinion Dynamics for Resource Allocation}

\author{Prashil Wankhede$^{1}$, Nirabhra Mandal$^{2}$, Sonia Mart\'{i}nez$^{2}$ and Pavankumar Tallapragada$^{1}$
  \thanks{ This work was partially supported by Science and Engineering Research Board under grant CRG/2023/008573.}  \thanks{$^{1}$Prashil Wankhede and Pavankumar Tallapragada are with the Indian Institute of Science. \{\tt\small \{pavant, prashilw\}@iisc.ac.in\}} \thanks{$^{2}$Nirabhra Mandal and Sonia Mart\'{i}nez are with the Department of Mechanical and Aerospace Engineering, University of California, San Diego.  \{\tt\small \{nmandal, soniamd\}@ucsd.edu\}} 
}

\begin{document}
\maketitle
\thispagestyle{empty} \pagestyle{empty}
\begin{abstract}
  We propose a model of opinion formation on resource allocation among multiple topics by multiple agents, who are subject to hard budget constraints. We define a utility function for each agent and then derive a projected dynamical system model of opinion evolution assuming that each agent myopically seeks to maximize its utility subject to its constraints. Inter-agent coupling arises from an undirected social network, while inter-topic coupling arises from resource constraints. We show that opinions always converge to the equilibrium set. For special networks with \emph{very weak antagonistic relations}, the opinions converge to a unique equilibrium point. We further show that the underlying opinion formation game is a potential game. We relate the equilibria of the dynamics and the Nash equilibria of the game and characterize the unique Nash equilibrium for networks with no antagonistic relations. Finally, simulations illustrate our findings.
\end{abstract}

\section{Introduction}
\label{sec:intro}

Multi-agent modeling and study of opinion dynamics finds widespread applications in sociology, economics, and other fields. Many of the complex decision making problems that individuals, organizations, corporations and governments engage in often boil down to problems of allocating limited resources (such as time or money) among multiple options. In such scenarios, there is a strong interplay between agents' self-interest, their heterogeneous resource constraints and social influences on the formation of opinions about resource allocation. To the best of our knowledge, currently there are no models and studies that explore such complex mechanisms behind opinion formation. Thus, in this paper, we seek to model and study such important phenomena.

\subsubsection*{Literature review}

Classical models, including the French-DeGroot (FD), the Abelson, the Taylor, the Friedkin-Johnsen (FJ), the Hegselmann-Krause (HK), the Altafini, and the Deffuant-Weisbuch (DW) models form the foundation of opinion dynamics research~\cite{2017_AP-RT_ARC,2018_AP-RT_ARC 
  ,2020_QZ_GK_HZ_HL_XC_CCL_YD}
and have inspired numerous extensions and refinements. Many of these models studied single-topic discussions with scalar opinions. Extensions to vector-valued opinions (multi-topic scenarios) have also been explored~\cite{2012_AN_BT,2017_SP_AP_RT_NF, 2020_MY_MT_YL_BA, 2022_GH_ZC_XW_MH,2023_GH_ZS_TH_WZ_XW,2020_HA_QT_MT_MY_JL_KM,2024_MR_CA}. 
Some notable works adopt a game-theoretic and utility-maximization approach to model opinion dynamics~\cite{2009_JM_GA_JS,2015_Bindel_et_al,2014_PG_JL_FS,2015_SRE-TB,2024_HJ_AY}. 
	
In many applications, opinions are naturally or forced to be constrained  within a bounded space. Models for such scenarios also exist. For instance,~\cite{2022_NV_DG} models opinions as unit vectors with varying orientations. References
~\cite{2020_MY_MT_YL_BA, 2017_SP_AP_RT_NF,2021_PCV_KC_FB} restrict opinions to a real interval, which is particularly useful when opinions represent an agent's sentiments. Another work~\cite{2015_MC_AC_BP} proposes a model where opinions are confined on the surface of a sphere.
	
The impact of limited or constrained resources on decision-making has been explored across various disciplines, including psychology, economics, and transportation. For a comprehensive overview, see~\cite{2019_RH_CM_AS_DT_VG,2015_CH_CM} and the references therein. Notably,~\cite{2019_RH_CM_AS_DT_VG} examines how financial constraints influence consumer behavior from four distinct angles: \emph{resource scarcity}, \emph{choice restriction}, \emph{social comparison}, and \emph{environmental uncertainty}. Motivated to study such effects, our recent works~\cite{2023_PW_NM_PT, 2025_PW_NM_SM_PT} introduced a nonlinear opinion dynamics model with a resource penalty in the utility function, enforcing soft constraints to keep opinions bounded. The applicability of the FJ model to resource allocation problems is examined in~\cite{2019_NF_AP_WM_FB}.
A closely related work is~\cite{2007_JL}, which introduces models (based on the DW and HK models) where opinions are confined to a probability simplex. Here, the opinions represent each agent's proposal for allocating a fixed global budget across multiple projects.
	
\subsubsection*{Statement of contributions}

\begin{enumerate}[leftmargin=0pt, labelsep=10pt, itemindent=18pt]

\item To the best of our knowledge, this is the first work to model and study opinion dynamics for general resource allocation scenarios. Our model captures heterogeneous goals and hard resource constraints of the agents. Moreover, we adopt a projected dynamical systems framework to model opinion dynamics under constraints on opinions, which is also a first in the literature.

\item We employ a game-theoretic framework, where the proposed projected dynamics are derived from the \emph{unconstrained} dynamics which in turn are derived from the agent's utility function, assuming that each agent myopically seeks to maximize its utility.

\item Unlike our prior works~\cite{2023_PW_NM_PT, 2025_PW_NM_SM_PT}, which focused on a single-topic model with soft resource penalty, this paper
 considers multiple topics and enforces hard resource constraints. Our model allows for agents with heterogeneous budgets.
\end{enumerate}
	
	\subsubsection*{Notation}
	The sets of natural numbers, real numbers, non-negative real numbers and positive real numbers are denoted by $\mathbb{N},\real$, $\real_{\geq0}$ and $\real_{>0}$, respectively.
	Given $\boldx, \boldy \in \real^{n}$ and a diagonal matrix 
	$\mathbf{D} \in \real_{\geq 0}^{n\times n}$; 
	$\|\boldx\| \ldef \sqrt{\boldx^\top \boldx}$, 
	$\|\boldx\|_{\boldD}:= \sqrt{\boldx^\top \boldD \boldx}$ and $\langle \boldx,\boldy \rangle \ldef \boldx^\top \boldy$. 
	For $\boldx \in \mathbb{R}^p$, $\support{\boldx}$ denotes its support, and $\mathrm{diag}(x_{1},\ldots,x_p)$ denotes a diagonal matrix with $\boldx$ as its main diagonal. Similarly, for $\boldD_1, \dots, \boldD_n\in \real^{p \times p}$,
	$\mathrm{blkdiag}(\boldD_1, \dots, \boldD_n)$ denotes a block diagonal matrix with $\boldD_i$ as its diagonal blocks. For real matrices $\mathbf{A}$, $\mathbf{B}$; their Kronecker product is $\mathbf{A} \otimes \mathbf{B}$ and $\|\mathbf{A}\|$ denotes the induced matrix norm. 
	Let $\boldzero$ denote a vector (of appropriate dimension) with all \emph{zero} elements. 
	For any $a \in \real$, $|a|$ denotes its absolute value. 
	For any $\boldx_{i} \in \realm$, $x_i^j$ denotes its $j\tth$ element. 

	For $\Mc{C}\subset \realn$, $\overline{\Mc{C}}$ is its closure.
	The empty set is denoted by $\varnothing$	.
	The projection of $\boldx \in \realn$ onto a closed, convex and non-empty set $\Mc{S}\subset \realn$ is denoted by $\Mc{P}_{\Mc{S}}(\boldx):=\argmin_{\boldy \in \Mc{S}} \|\boldy - \boldx\|$. The tangent cone to a set $\Mc{X} \subset \realn$ at a point $\boldx \in \Mc{X}$, $T_{\Mc{X}}(\boldx) \ldef \{ \boldv \in \realn \,|\,$ $\exists \{\boldx_{i}\}_{i \in \mathbb{N}} \in \Mc{X}, \boldx_i \to \boldx$ and $ \exists \{t_{i}\}_{i \in \mathbb{N}}, t_i >0, t_i \to 0, $ such that $\boldv = \lim_{i \to \infty}(\boldxi - \boldx)/t_i \}$.

	\section{Preliminaries}
	\label{sec:preliminaries}	

A Projected Dynamical System (PDS) is a special type of discontinuous dynamical system that restricts the state to a constraint set. In this work, we limit the definition of a PDS to closed, convex constraint sets and Euclidean projections. For a given closed, convex, non-empty set $\Mc{X} \subset \real^p$ and a vector field $\boldg: \real^p \to \real^p$, we define the PDS as
	\begin{equation*}
		\dot{\boldx} = \Pi_{\Mc{X}}(\boldx,-\boldg(\boldx)) \ldef \lim_{\delta \to 0^+} \frac{ \Mc{P}_{\Mc{X}}(\boldx - \delta \boldg(\boldx))-\boldx}{\delta} \,,
	\end{equation*}
	which is equivalent to the following dynamics~\cite{2020_BB_AT},
	\begin{equation}
		\dot{\boldx} = \Mc{P}_{T_{\Mc{X}}(\boldx)}(-\boldg(\boldx)) \,.
		\label{eq:pds_def}
	\end{equation}
	In this paper, we use this representation of the PDS. Since the dynamics~\eqref{eq:pds_def} is discontinuous, it is essential to first define an appropriate notion of a solution before proceeding with the analysis. Here, we consider solutions of~\eqref{eq:pds_def} in the \emph{Carath\'{e}odory} sense~\cite{2008_JC} \emph{i.e.,} 
	absolutely continuous curves that satisfy~\eqref{eq:pds_def} almost everywhere. 
	It is well known that the equilibrium points of \eqref{eq:pds_def} are related to solutions of the following variational inequality (VI) problem~\cite{1996_AN_DZ_book}.
		\begin{define}{(\textbf{\emph{Variational Inequality Problem}}~\cite{1996_AN_DZ_book})}
                  For a closed, convex set $\Mc{X} \subseteq \real^p$ and a function $\boldg:\Mc{X} \to \real^p$, the finite dimensional variational inequality problem, $\mathrm{VI}(\boldg,\Mc{X})$, is to find $\boldx^* \in \Mc{X}$ such that $ \langle \mathbf{g}(\boldx^*), \boldx - \boldx^* \rangle \geq 0, \quad\forall \boldx \in \Mc{X}.$ We refer to the solution set for $\mathrm{VI}(\boldg,\Mc{X})$ by $\solvarineq$. \bulletend
	\end{define}
	
Interestingly, the VI problem is closely connected to several fundamental mathematical frameworks, including nonlinear equations, optimization, complementarity problems, and fixed-point problems~\cite{1996_AN_DZ_book}.
\section{Modeling and Problem Setup}
\label{sec:modeling_and_problem_setup}
Consider a set $\agt:=\{1,\ldots,n\}$ of $n$ agents that form opinions about allocation of their limited resources on $m$ topics indexed by the set $\tp:=\{1,\ldots,m\}$.  We aim to study the evolution of these opinions as the agents interact with each other over a social network. 
We start by defining a utility function for each agent that depends on their preferences and opinions on all topics and the social influence of their neighboring agents. Since the agents' opinions are about allocation of limited resources among the topics in $\tp$, they constrain their opinions so that a budget constraint is always satisfied. We first derive the \emph{unconstrained} opinion dynamics from the utility functions, assuming that each agent myopically seeks to maximize its utility. Then, using the unconstrained vector field, we define a projected dynamical system that keeps the opinions of all agents restricted by their budget constraints.
	
We denote the opinion of agent $i \in \agt$ on topic $j \in \tp$ at time $t$ as $\zij(t) \in \real$. We omit the time argument wherever there is no confusion. We let $\boldzi \ldef [\zi^{1},\ldots,\zi^{m}]^\top \in \realm$ denote the stacked opinions of agent $i$ on all topics. Similarly, $\boldz \ldef [\boldz_{1}^\top,\ldots,\boldz_{n}^\top]^\top \in \realmn$ is the stacked vector of opinions of all agents. The utility of agent $i \in \agt$ is
\begin{align}
  U_{i}(\boldz) & \ldef  - \frac{1}{2} \tpsumj \wij \big(\zij - \dij
    \big)^{2} - \frac{1}{2} \tpsumj \agtsumk \aik \big(\zkj - \zij
    \big)^{2} \nonumber \\
  &= -\frac{1}{2} || \boldzi-\boldpi ||_{\boldDi}^2
  - \agtsumk \frac{\aik}{2} || \boldzk-\boldzi ||^2.
                  \label{eq:utility}
\end{align}
Here $\dij \in \real$ is the internal \emph{preference} of agent $i \in \agt$ on topic $j \in \tp$, while $\boldpi \in \realm$ is the vector containing preferences $\dij$'s of agent $i$ on all topics. The \emph{topic preference} weight $\wij \in \real_{> 0}$ is the importance that agent $i \in \agt$ attaches to its internal preference on topic $j \in \tp$ and $\boldDi \ldef \mathrm{diag}(\wi^{1},\ldots,\wi^{m}) \in \real^{m \times m}$.
Finally, $a_{ik}$ is the \emph{influence weight} of agent $k$'s opinions on those of agent $i$. Let $\adjmat$ be the matrix with its $(i, k)^{\text{th}}$ element as $a_{ik}$. The inter-agent influence weights implicitly define a weighted social graph $\grph\ldef \big(\agt, \edg, \adjmat \big)$, where $\agt$ (the set of agents) is the set of nodes, $\edg \subseteq \agt \times \agt$ is the set of undirected links and $\adjmat$ is the weighted adjacency matrix. Note that $\aik \neq 0$ if and only if $\{i,k\} \in \edg$. Notice that the self loop weights $a_{ii}$, for any $i \in \agt$, do not affect the utility of the agents. We may assume them to be zero without loss of generality. The self influence of the agents is captured by the preference term in the utility function. In this paper, an agent's opinions are restricted to a constraint set which we define next. The opinion constraint set for each agent $i \in \agt$ is given by
\begin{equation}
  \label{eq:constraint_set}
  \constraintseti \ldef \{\boldzi \in \realm \:\mid\: \ci^\top \boldzi \leq \Bi,\:\boldzi \geq \boldzero \},
\end{equation}
where $\ci \in \realm_{>0}$ is a weighing parameter against $i$'s budget $B_i$. Note that $\constraintseti$ is a closed, bounded and convex polytope. 

We assume the following throughout the paper.  
\begin{enumerate}[label=\textbf{(SA\arabic*)},wide=\parindent]
\item \thmtitle{Undirected social network}
  \label{asmp:undirected_graph}
  The graph $\mathsf{G}=\big(\agt, \edg, \adjmat \big)$ is undirected, \emph{i.e.,} $\adjmat=\adjmat^\top$ and connected\footnote{Graph connectedness is without loss of generality, as each connected component can be analyzed separately in isolation.}. \remend

\item \thmtitle{Model parameters' signs} $\dij \in \real_{\geq 0}$, $\wij \in \real_{> 0}$, $\aik \in \real$, $\forall k \in \agt$, $B_{i}>0$ and $\ci > \boldzero$, $\forall i \in \agt$, $j \in \tp$. \label{asmp:param_sign} \remend
\end{enumerate}

We discuss the justification for \ref{asmp:param_sign} in the sequel. However, note here that the entries of $\adjmat$ are allowed to have negative values. The sign of $\aik$ denotes the type of influence relationship and its magnitude denotes the degree of influence. Using this idea, we define the following.
	\begin{define} \label{def:neighbor_set}
		\thmtitle{Neighbor sets of an agent}
		For each agent $i \in \agt$, the set $\Nie \ldef \{k \in 
		\agt\setminus\{i\} \:|\: \aik<0 \}$ denotes the set of its 
		\emph{enemies} and the set $\Nif \ldef \{k \in \agt\setminus\{i\} 
		\:|\: 
		\aik > 0\}$ denotes the set of its \emph{friends}. Further, $\Ni \ldef \Nie \cup \Nif$	 
		denotes the set of \emph{neighbors} of agent $i \in \agt$.
		\bulletend
	\end{define}
Some of our results in this paper are applicable to certain specific classes of inter-agent interactions. We list these specific classes in the following assumptions.
	
	\begin{enumerate}[label=\textbf{(A\arabic*)},wide=\parindent] 
		\item\thmtitle{No antagonistic relations} For each agent $i \in \agt$, $\Nie=\varnothing$.
		\label{asmp:no_antagonistic_relations}
		\remend
	\end{enumerate}	
	
	\begin{enumerate}[resume, label=\textbf{(A\arabic*)},wide=\parindent] 
		\item \thmtitle{Very weak antagonistic relations} For each agent $i \in \agt$,
		$ \min_{j \in \tp}\;\wij > 2 \enemysumi |\aik|$. \bulletend \label{asmp:WAR_WSI}
	\end{enumerate} 
	
	\begin{enumerate}[resume,label=\textbf{(A\arabic*)},wide=\parindent] 
		\item\thmtitle{Weak antagonistic relations}
		For each agent $i \in \agt$,
		$ \big(\min_{j \in \tp} \:\wij + \agtsumknoti \aik \big) \geq 0.$ \bulletend
		\label{asmp:WAR}
              \end{enumerate}
	
\begin{remark}\thmtitle{On the assumptions about the inter-agent relations}
 Under Assumption~\ref{asmp:WAR_WSI} the aggregate influence of agent $i$'s enemies is ``very weak" compared to its preference weights. Assumption~\ref{asmp:WAR} means that the aggregate influence weights of agent $i$'s enemies is ``weak" and insufficient to counter its preference weights and the influence of its friendly neighbors.
 Note that as a consequence of Standing Assumption~\ref{asmp:param_sign}, \ref{asmp:no_antagonistic_relations} $\implies$ \ref{asmp:WAR_WSI} $\implies$ \ref{asmp:WAR}.  \bulletend
\end{remark}
	
\begin{remark}
\thmtitle{Opinions, Budget and Utility}\label{rem:example}
In this work, we propose an opinion dynamics model for resource allocation scenarios under hard budget constraints. The resources could represent such things as time, money, or computational resources. As a motivating example, consider the scenario of agents budgeting their limited money for expenses on various needs and wants, such as food, housing, means of transport and digital devices to buy. In our model, $\zij \in \real_{\geq 0}$ is agent $i$'s opinion of how much to allocate to topic or option $j$; whereas the preference $\dij \geq 0$ represents agent $i$'s internal preference for allocation on topic $j$. The weight $\wij$ represents the importance that the agent gives or the confidence the agent has in its internal preference. Either because the allocation problem is a complex one or due to other reasons, an agent might mimic the opinions of its social neighbors. The weight $\aik$ quantifies agent $i$'s level of trust or distrust in agent $k$'s opinion about any topic $j$. The cost per unit allocation by agent $i$ for option $j$ is captured by $c_i^j$. Consequently, $c_i^j \zij$ denotes the portion of agent $i$'s budget allocated to option $j$, and $\boldci^\top\boldzi$ represents the total allocation, which is constrained by its budget $B_i$.
\remend
\end{remark}
Remark~\ref{rem:example} justifies our standing assumption~\ref{asmp:param_sign} on all model parameters stated earlier.
In the following, we first define the \emph{unconstrained} multi-topic opinion dynamics and then use it to propose the projected opinion dynamics model.

\subsubsection*{\textbf{Unconstrained Opinion Dynamics}}

Suppose that there are no budget constraints and opinions $\zij$'s are allowed to take any real values. We assume that at each infinitesimal time instant, each agent $i \in \agt$ revises its opinion on topic $j \in \tp$ by ascending along the gradient of its utility $U_{i}$, given in~\eqref{eq:utility}, with respect to its own opinion $\zij$, while assuming that others do not change their opinions. Thus, for each $i \in \agt$,
	\begin{align}
		\label{eq:uncontrained_agent_dynamics}
		\hspace{-0.15cm}\dot{\boldz}_i = -\boldDi(\boldzi-\boldpi) + \agtsumk \aik (\boldzk-\boldzi) \rdef -\boldfi(\boldz_{i}, \boldz_{-i}).
	\end{align}	
The unconstrained dynamics of all agents can be written as	
	\begin{align}
		\label{eq:uncontrained_dynamics}
		\dot{\boldz} = -\boldD(\boldz - \boldp) - (\laplacian \otimes \identitym)\boldz \rdef -\boldf(\boldz),
	\end{align}
	where $\boldznoti \in \real^{(n-1)m}$ denotes the opinions of all agents except $i$, $\boldD \ldef \mathrm{blkdiag}(\boldD_{1},\ldots,\boldD_{n}) \in \real^{mn \times mn}$, $\boldp := [\boldp_{1}^\top,\ldots,\boldp_{n}^\top]^\top \in \realmn$ and $\laplacian$ is a symmetric matrix\footnote{In the general case, $\laplacian$ is referred to as the repelling signed Laplacian matrix~\cite{2025_SZ_YT_WC}. 
	If $\aik \geq 0,\forall i,k \in \agt$ then $\laplacian$ denotes the usual (positive semi-definite) graph Laplacian matrix.} with elements $l_{ii}=\agtsumknoti\aik$ and $l_{ij} = -\aij,\: \forall j\neq i$. We denote the Jacobian of the vector field $\boldf(\boldz)$ in~\eqref{eq:uncontrained_dynamics} by,
	\begin{align}
			\boldJ:=(\boldD+\laplacian\otimes\identitym).
			\label{eq:jacobian_mat}
		\end{align}

	The terms in the unconstrained dynamics~\eqref{eq:uncontrained_agent_dynamics} for each topic is similar to Taylor's model~\cite{1968_Taylor} of opinion dynamics when antagonistic relations are absent in the network (\emph{i.e.,} when $\aik\geq 0$). In the presence of antagonistic relations, the second term in~\eqref{eq:uncontrained_agent_dynamics} models the \emph{boomerang effect}~\cite{1967_RA_JM_boomerang_effect}, where agents respond to neighbors' influence by shifting their opinions in the opposite direction, leading to disagreement.

	Since the unconstrained dynamics in~\eqref{eq:uncontrained_agent_dynamics} is obtained by performing gradient ascent on each agent~$i$’s utility with respect to its opinion~$\boldzi$, it can be interpreted in continuous time as the instantaneous better-response dynamics of agent~$i$ to its neighbors' opinions. However, because each agent’s opinions are restricted by individual resource constraints given in~\eqref{eq:constraint_set}, the unconstrained dynamics~\eqref{eq:uncontrained_dynamics} cannot, in general, ensure feasibility of the opinions at all times. So, we propose a projected version of the opinion update rule.

\subsubsection*{\textbf{Projected Opinion Dynamics}}
In this work, our aim is to study the effect of limited resources of the agents on their opinion evolution. Hence, we propose a model where the opinions of each agent satisfy their resource constraints~\eqref{eq:constraint_set} at all times. More precisely, the projected dynamical system
	\begin{equation}
		\label{eq:projected_dynamics_agent}
		\dot{\boldz}_i = \projtki(-\boldfi(\boldz_{i}, \boldz_{-i})) \,,
	\end{equation}	
	uses the unconstrained gradient ascent direction in~\eqref{eq:uncontrained_agent_dynamics} and restricts each agent's opinion evolution within the feasible set perpetually. Again for the
	sake of convenience, we let $\znoti \in \constraintsetnoti:=\bigtimes_{l=1,l\neq i}^{n}\constraintset_{l}$ denote the opinions of all agents except $i$.	 
	The tangent cone $\tki$ to the constraint set $\constraintseti$ at $\boldzi$ is closed and convex because the constraint set $\constraintseti$ in~\eqref{eq:constraint_set} is closed and convex~\cite{1998_RTR_RJ}. Hence, the projection in~\eqref{eq:projected_dynamics_agent} is unique~\cite{2003_FF_PJ_book} and consequently, \eqref{eq:projected_dynamics_agent} takes the form of a differential equation rather than a differential inclusion.
	Similarly, we can associate the following projected dynamical system with~\eqref{eq:uncontrained_dynamics}
	\begin{equation}
		\label{eq:projected_dynamics}
		\dot{\boldz} = \projtk(-\boldf(\boldz)), \quad\text{where}\: \constraintset \ldef \textstyle\bigtimes_{i=1}^{n}\constraintseti. 
	\end{equation}
	Note that, the inter-agent opinion coupling in~\eqref{eq:projected_dynamics_agent} comes through the social term in~\eqref{eq:utility} and the inter-topic opinion coupling comes through the budget constraints given in~\eqref{eq:constraint_set}. 
	
	The equivalence between the agents' dynamics in~\eqref{eq:projected_dynamics_agent} and~\eqref{eq:projected_dynamics} can be established as follows. 
	First, note that each agent $i$'s constraint set $\constraintseti$ defined in~\eqref{eq:constraint_set} is closed and convex and therefore \emph{Clarke}-regular~\cite{1998_RTR_RJ}. Since $\constraintset = \textstyle\bigtimes_{i=1}^{n}\constraintseti$ by definition~\eqref{eq:projected_dynamics} and each $\constraintseti$ is Clarke-regular, $\constraintset$ is Clarke-regular and hence, the tangent cone decomposes as $\tk = \bigtimes_{i=1}^{n} \tki,\:\forall \boldz \in \constraintset$~\cite{1998_RTR_RJ}. Second, the projection onto the tangent cone in both~\eqref{eq:projected_dynamics_agent} and~\eqref{eq:projected_dynamics} is performed using the Euclidean norm. These two facts ensure that the two formulations describe the same opinion evolution. 
	
	For simplicity, we will refer to dynamics~\eqref{eq:projected_dynamics_agent} and~\eqref{eq:projected_dynamics} by $\pdsi$ and $\pds$ respectively. Finally, we denote the equilibrium points of $\pds$ by
	\begin{equation}
		\label{eq:eqpt}
		\eqpt \ldef \{\boldz \in \constraintset \:\mid\: \dot{\boldz}= \projtk(-\boldf(\boldz)) = \boldzero \}.
	\end{equation}
	
	We now show that the projected dynamics~\eqref{eq:projected_dynamics_agent} is also an instantaneous better response dynamics. Specifically, we demonstrate that the projected update direction of each agent $i \in \agt$ in~\eqref{eq:projected_dynamics_agent} aligns with the unconstrained gradient ascent direction in~\eqref{eq:uncontrained_agent_dynamics}. The proof (which we omit) of this result is based on~\cite[Lemma 2.1]{1996_AN_DZ_book}.
		\begin{lemma}
			\thmtitle{$\pdsi$ is an instantaneous better response dynamics} Consider the agent dynamics~\eqref{eq:projected_dynamics_agent} and the unconstrained dynamics~\eqref{eq:uncontrained_agent_dynamics} for any agent $i\in \agt$ . Then for any $\boldz \in \constraintset$,
			$-\boldfi(\boldzi,\boldznoti)^\top \projtki\left(-\boldfi(\boldzi,\boldznoti)\right) \geq 0.$ \remend
		\end{lemma}

We are now ready to state the goals of this paper.
	
\subsubsection*{Objectives} For the proposed projected opinion dynamics~\eqref{eq:projected_dynamics}, our primary objective is to study how inter-agent social influence and inter-topic coupling, through agents' budget constraints, shape opinion formation across different topics. We aim to analyze the long-term evolution of opinions, particularly the convergence and stability properties of the dynamics in~\eqref{eq:projected_dynamics}. We also aim to investigate the properties of Nash equilibria in the underlying opinion formation game.
	
\section{Asymptotic Behavior of Opinions}
\label{sec:asymptotic_analysis}
In the following result, we first leverage the affine structure of the unconstrained vector field~\eqref{eq:uncontrained_dynamics} to establish the existence and uniqueness of the Carath\'{e}odory solution to~\eqref{eq:projected_dynamics}.

	\begin{theorem}
		\thmtitle{Existence and Uniqueness of Carath\'{e}odory solution}
		Consider the dynamics $\pds$ defined in~\eqref{eq:projected_dynamics}. Then, there exists a unique Carath\'{e}odory solution $\boldz(t),\forall t \geq 0$ of $\pds$ starting from any feasible initial condition $\boldz(0)\in \constraintset$.
	\end{theorem}
	\begin{proof}
		To prove this claim, we first note the following for the vector field $\boldf(\boldz)$ defined in~\eqref{eq:uncontrained_dynamics},
		$\|\boldf(\boldz)\| 
			\leq \alpha (1+\|\boldz\|) \:,\forall \boldz \in \constraintset$
		where $\alpha := \max \{\|\boldD+(\laplacian \otimes \identitymn)\|,\| \boldD\boldp \| \}$. For any $\boldz_1,\boldz_2 \in \constraintset$, we have
		$\left[\boldf(\boldz_2)- \boldf(\boldz_1)\right]^\top [\boldz_1-\boldz_2]
			\leq  \alpha \|\boldz_1-\boldz_2\|^2.$
		The claim now follows from \cite[Theorem 2.5]{1996_AN_DZ_book}.
	\end{proof}
	
	\vspace{5pt}
	We now define the variational inequality associated with each agent's dynamics~\eqref{eq:projected_dynamics_agent} and establish a connection between the solutions of these VI's and the single VI associated with~\eqref{eq:projected_dynamics}. We then relate the solution of this VI to the equilibrium points~\eqref{eq:eqpt} of the projected dynamics~\eqref{eq:projected_dynamics}.
		
\begin{define}
  \thmtitle{Agent specific Variational Inequality Problem} For any agent $i \in \agt$, given a closed, convex set $\constraintseti\subseteq \realm$ and a function $\boldfi(\cdot,\boldznoti):\constraintseti \to \realm$ for some $\boldznoti \in \constraintsetnoti$, the agent specific variational inequality problem is to find $\boldzi^* \in \constraintseti$ such that $ \langle \boldfi(\boldzi^{*},\boldznoti), \boldx_{i} - \boldzi^{*} \rangle \geq 0, \quad\forall \boldxi \in \constraintseti.$ We refer (with some abuse of notation) to this problem and its solution set by $\varineqfi$ and $\solvarineqfi$, respectively. \bulletend
	\end{define}	

\begin{lemma}
  \thmtitle{Relation between solutions of $\varineqfi$ and $\varineqf$}
  $\boldz^* = (\boldzi^*,\boldznoti^*) \in \solvarineqf$ if and only if $\boldzi^* \in \solvarineqfistar,\:\forall i \in \agt$.
		\label{lem:reln_betn_var_ineq}
	\end{lemma}
	\begin{proof}
		To prove this, observe that for any $\boldx,\boldz \in \constraintset$
		\begin{equation}
			\innerproduct = \textstyle\agtsumi \: \innerproducti.
			\label{eq:relation_betn_varineq_proof_1}
		\end{equation}
		Suppose $\boldz^* \in \solvarineqf$. Using~\eqref{eq:relation_betn_varineq_proof_1}, we have
		\begin{equation}
			\textstyle\agtsumi \,\, \innerproductistar \geq 0, \: \forall \boldx \in \constraintset.
			\label{eq:relation_betn_varineq_proof_2}
                      \end{equation}
Let $\boldy^{(i)}:=  [\boldz_{1}^{*\top},\ldots,\boldz_{i-1}^{*\top},\bolds_{i}^\top,\boldz_{i+1}^{*\top},\ldots,\boldz_{n}^{*\top}]^\top \in \constraintset$, $\forall \bolds_i \in \constraintseti$ and $\forall i \in \{ 1,\ldots,n \}$.            
With $\boldx = \boldy^{(i)}$, \eqref{eq:relation_betn_varineq_proof_2} gives
		\begin{align*}
			0 & \leq \langle \boldfi(\boldzi^*,\boldznoti^*),\bolds_i-\boldzi^* \rangle + \textstyle\sum_{j \neq i} \langle \boldf_j(\boldzj^*,\boldznotj^*),\boldzj^*-\boldzj^* \rangle \\
			  & = \langle \boldfi(\boldzi^*,\boldznoti^*),\bolds_i-\boldzi^* \rangle, \quad \forall \bolds_i \in \constraintseti\,. 
		\end{align*}
		Hence, $\boldzi^* \in \solvarineqfistar,\:\forall i \in \agt$.
		
		Next, suppose $\boldzi^* \in \solvarineqfistar \:, \forall i \in \agt$. Then $\forall i \in \agt$, 
		$
			\innerproductistar \geq 0,\: \forall \boldxi \in \constraintseti.
		$
Then, ~\eqref{eq:relation_betn_varineq_proof_1} implies
		\begin{align*}
			0 \leq  \textstyle\sum\limits_{i \in \agt} \innerproductistar = \langle \boldf(\boldz^*),\boldx-\boldz^* \rangle, \,\, \forall \boldx \in \constraintset\,. 
		\end{align*}		 
Hence, $\boldz^*\in \solvarineqf$. The proof is now complete.
	\end{proof}
		
	In the next result, we use a well known result from~\cite{1996_AN_DZ_book} to 
relate $\eqpt$, the equilibrium points of $\pds$, and the solution set of $\varineqf$ and show that $\eqpt$ is not empty.
	
\begin{lemma}
  \thmtitle{Existence of equilibrium opinions}
  \label{lem:eqpt_exist}
  For the dynamics $\pds$ in~\eqref{eq:projected_dynamics},
  $\eqpt = \solvarineqf \neq \varnothing$.
\end{lemma}
\begin{proof}
  The claim that $\eqpt = \solvarineqf$ is a direct consequence of~~\cite{1996_AN_DZ_book}. Next, since $\constraintseti$ is compact and convex, $\constraintset$ as defined in~\eqref{eq:projected_dynamics} is also compact and convex. Further, the vector field $\boldf(\boldz)$ is affine and hence continuous in $\constraintset$. Non-emptiness of $\eqpt = \solvarineqf$  now follows from~\cite[Theorem 2.1]{1996_AN_DZ_book}.
\end{proof}
		
We can immediately show that any opinion trajectory converges to the set of equilibrium points, which we know is non-empty (due to Lemma~\ref{lem:eqpt_exist}). First, we define the function
\begin{align}
	W(\boldz) & := \agtsumk U_{k}(\boldz) + 0.5 \,\boldz^\top(\laplacian\otimes \identitym)\boldz \notag \\
	& = - 0.5 \, \boldz^\top \boldJ \boldz + \boldz^\top \boldD\boldp - 0.5 \, \boldp^\top\boldD\boldp \,,
	\label{eq:potential_function}
\end{align}
which will be useful for showing convergence. Later, in Section~\ref{sec:game_theoretic_analysis}, we show that this function can work as a potential function for the associated non-cooperative game.
	
		\begin{theorem}
		\thmtitle{Opinions tend to the equilibrium set}
		Consider the dynamics $\pds$ in~\eqref{eq:projected_dynamics}. Let $\boldz(t)$ denote the unique solution to $\pds$ starting from a feasible initial condition $\boldz_0 \in \constraintset$. Then, $\boldz(t)$ converges to $\eqpt$.
		\label{thm:convg_to_eqm_set}
		\end{theorem}
		\begin{proof}
		Consider the Lyapunov-like function $V(\boldz) \ldef -W(\boldz)$, with $W(\cdot)$ defined in~\eqref{eq:potential_function}.
		Then, it can be verified that $\nabla V(\boldz)=\boldf(\boldz)$. The time derivative of $V(\boldz)$ along the trajectories of~\eqref{eq:projected_dynamics} is $ \dot{V}(\boldz) = \boldf(\boldz)^\top \projtk(-\boldf(\boldz)) \leq 0$.  This inequality can be proved using~\cite[Lemma 2.1]{1996_AN_DZ_book} and Cauchy-Bunyakovskii-Schwarz inequality.
                Now, let us define the set $\mathcal{W}:=\{\boldz \in \constraintset \;\mid \; \dot{V}(\boldz)=0\}$.  Using~\cite[Lemma 2.1]{1996_AN_DZ_book}, we can again conclude that $\dot{V}(\boldz)=0$ if and only if $\boldf(\boldz)=\boldzero$ or $\projtk(-\boldf(\boldz))=\boldzero$. Since $\{\boldz \in \constraintset\;\mid\; \boldf(\boldz)= \boldzero \} \subseteq \eqpt$, we conclude that $\mathcal{W}=\eqpt$. Notice that the vector field $\boldf(\boldz)$ defined in~\eqref{eq:uncontrained_dynamics} is affine. Hence, the variational inequality associated with the corresponding $\pds$ is an \emph{Affine Variational Inequality}. It is well known that the solution set of an Affine variational inequality is closed~\cite[Corollary 5.3]{2005_GL_NT_NY_book}. This, along with Lemma~\ref{lem:eqpt_exist}, implies that $\eqpt$ is closed. Thus, $\overline{\mathcal{W}}=\eqpt$. Finally, we apply the invariance principle for discontinuous dynamics~\cite[Theorem 2]{2008_JC} to ensure convergence of $\boldz(t)$ to the largest invariant set in $\overline{\mathcal{W}}$, which in our case is the equilibrium set $\eqpt$ itself.
              \end{proof}
	
In the remaining portion of this section, we leverage the properties of the Jacobian Matrix, $\boldJ$, of the unconstrained vector field $\boldf(\boldz)$ to comment more about the equilibrium set and the asymptotic behavior of the opinion trajectories.
In fact, if $\boldJ$ is positive semi-definite, we can establish that $\eqpt$ is convex and show convergence of solutions of~\eqref{eq:projected_dynamics} to an equilibrium point in $\eqpt$. We formally state these results next. 

\begin{theorem}
  \thmtitle{Properties of equilibria and opinion trajectories for positive semi-definite Jacobian}
  \label{thm:inv_jacobian_properties}
  Consider the dynamics $\pds$ in~\eqref{eq:projected_dynamics}, its equilibria $\eqpt$ in~\eqref{eq:eqpt} and $W(\cdot)$ in~\eqref{eq:potential_function}. Let $\boldz(t)$ be the unique solution to $\pds$ from a feasible initial condition $\boldz_0 \in \constraintset$. Moreover, let $\boldJ$ be positive semi-definite.
  Then the following are true.
  \begin{enumerate}
			\item \label{lem:global_monotone_attractivity_and_stability} Every equilibrium point $\boldz^*\in \eqpt$ is a global monotone attractor\footref{footnote:monotonicity_and_stability_defs} and hence Lyapunov stable.
			
			\item \label{lem:convexity_of_eqm_set} $\eqpt = \argmax_{\boldz \in \constraintset} W(\boldz)$, and hence $\eqpt$ is convex.
			
			\item \label{thm:convg_to_eqm_pt} $\boldz(t)$ converges to the equilibrium point $\boldz_0^* \in \eqpt$, where $\boldz_0^{*} := \argmin_{\boldy \in \eqpt} \|\boldz_0-\boldy\|$.
		\end{enumerate}
	\end{theorem}
	\begin{proof}
	To prove Claim~\ref{lem:global_monotone_attractivity_and_stability}, notice that under the stated assumptions the vector field $\boldf(\boldz)$ defined in~\eqref{eq:uncontrained_dynamics} is monotone\footref{footnote:monotonicity_and_stability_defs} on $\constraintset$. The claim now follows from~\cite[Theorem 3.5]{1996_AN_DZ_book}.
	
Next, to prove Claim~\ref{lem:convexity_of_eqm_set}, using the hypothesis, note that $-W(\boldz)$ is convex on $\constraintset$, since its Hessian $\boldJ$ is positive semi-definite in $\constraintset$. From~\cite[Proposition 2.3]{1996_AN_DZ_book}, we know that the set $\solvarineqf$ is equal to the set of solutions to the convex optimization problem $\max_{\boldz \in \constraintset} W(\boldz)$. Hence, $\solvarineqf$ is convex. The claim now follows from Lemma~\ref{lem:eqpt_exist}.
	
Finally, we show Claim~\ref{thm:convg_to_eqm_pt}. We know from Claim~\ref{lem:convexity_of_eqm_set} and~\cite[Corollary 5.3]{2005_GL_NT_NY_book} that $\eqpt \subseteq \constraintset$ is convex and compact. This means for any given $t \geq 0$, $\argmin_{\boldy \in \eqpt} \|\boldz(t)-\boldy\|$ is a singleton set. We now show that $\boldz(t)$ converges to $\boldz^*_{0}$. From Claim~\ref{lem:global_monotone_attractivity_and_stability}, we know that $\boldz^{*}_{0}$ is a global monotone attractor which means $\|\boldz(t)-\boldz^*_0\|$ is non-increasing, $\forall t \geq 0$. We also know that $\boldz(t)$ is continuous, $\forall t \geq 0$. All of the above observations imply that $\argmin_{\boldy \in \eqpt}\|\boldz(t)-\boldy\|=\boldz^*_0\:,\forall t \geq 0$. Additionally, from Theorem~\ref{thm:convg_to_eqm_set}, we know that $\boldz(t)$ converges to $\eqpt$, \emph{i.e.,} $\min_{\boldy \in \eqpt}\|\boldz(t)-\boldy\| \to 0$ as $t \to \infty$. Hence, $\|\boldz(t)-\boldz^*_0\|\to 0$ as $t \to \infty$. The proof is now complete.
\end{proof}

Note that Theorem~\ref{thm:inv_jacobian_properties} precludes the existence of periodic or oscillatory solutions for the dynamics~\eqref{eq:projected_dynamics} under the assumption that $\boldJ$ is positive semi-definite.
To conclude this section, we provide a stricter result for networks with very weak antagonistic relations~Assumption~\ref{asmp:WAR_WSI}.
In such a scenario, the following result guarantees convergence of the solution $\boldz(t)$ to a unique equilibrium point, starting from any initial condition $\boldz(0) \in \constraintset$.	
	
\begin{theorem}
  \thmtitle{Convergence to a unique equilibrium point for very weak antagonistic relations} Consider the dynamics $\pds$ defined in~\eqref{eq:projected_dynamics}. Suppose Assumption~\ref{asmp:WAR_WSI} holds. Then $\eqpt$ is a singleton set. Moreover, the unique equilibrium point is a strictly global monotone attractor\footref{footnote:monotonicity_and_stability_defs}.
  \label{thm:convg_to_uniq_eqm}
\end{theorem}
\begin{proof}
  We will show that under Assumption~\ref{asmp:WAR_WSI}, $\boldf(\boldz)$ is strictly monotone\footnote{We refer the readers to~\cite{1996_AN_DZ_book} for the definitions and further details.\label{footnote:monotonicity_and_stability_defs}} on $\constraintset$. Firstly, note that the Jacobian matrix $\boldJ$ defined in \eqref{eq:jacobian_mat} is symmetric. Thus, $\boldf(\boldz)$ is strictly monotone on $\constraintset$ if and only if $\boldJ$ has all positive eigenvalues. For any fixed $i\in \agt$ and $j \in \tp$, the center and radius of the corresponding Gerschgorin discs of $\boldJ$ are given by $\wij + \agtsumknoti \aik$ and $\agtsumknoti |\aik|$, respectively.  Using the Gerschgorin disc theorem, it can now be easily verified that the matrix $\boldJ$ is positive definite under Assumption~\ref{asmp:WAR_WSI}. The rest of the claim now follows from Lemma~\ref{lem:eqpt_exist},~\cite[Theorem 2.2]{1996_AN_DZ_book} and~\cite[Theorem 3.6]{1996_AN_DZ_book}.

\end{proof}

\section{Opinion Formation Game}
\label{sec:game_theoretic_analysis}
	
	In this section, we delve deeper into the properties of the underlying opinion formation game.
	Recall that each agent $i \in \agt$ is interested in maximizing its utility
	$U_i$ given in \eqref{eq:utility}, by suitably choosing its opinion
	$\boldzi \in \constraintseti$. This framework defines a \emph{strategic form game} $
		\mathcal{G}:=\langle \agt, \left(\constraintseti\right)_{i \in \agt}, 
	\left(U_{i}\right)_{i \in \agt}\rangle
	$
	among the set of agents
	$\agt$, with the \emph{strategy} of agent $i \in \agt$ being its opinion
	$\boldzi \in \constraintseti$ and its utility function being $U_i$.
	
	For this game, we are interested in presenting the properties of the Nash equilibria and relating it to the equilibrium set of the opinion dynamics. Note that, by definition, for a Nash equilibrium $\boldz^*$ (if one exists), $\boldzi^*$ is agent $i$'s best response over all opinions $\boldzi \in \constraintseti$ to $\boldznoti^* \in \constraintsetnoti$, the Nash equilibrium opinions of all the other agents. The set of \emph{Nash equilibria} of the game $\mathcal{G}$ is
	\begin{align}
		\notag \ne & \ldef \{ \boldz^* \in \constraintset \,\,|\,\, \forall i \in \agt, \\
&\quad U_i(\boldzi^*,\boldznoti^*) \geq U_i(\boldzi,\boldznoti^*), \forall \boldzi \in \constraintseti \}\,. \label{eq:ne}
	\end{align}

We start by relating the equilibria of the opinion dynamics with the Nash equilibria of the opinion formation game. 
	
	\begin{theorem}
          \thmtitle{Relation between $\eqpt$ and $\ne$} Consider the equilibrium points $\eqpt$ in~\eqref{eq:eqpt} of the dynamics $\pds$ in~\eqref{eq:projected_dynamics} and the set of Nash equilibria $\ne$ in~\eqref{eq:ne} of the game $\mathcal{G}$. Then $\ne \subseteq \eqpt$. 
		 
		Moreover, if Assumption~\ref{asmp:WAR} holds, then $\ne = \eqpt$.
		\label{thm:reln_betn_eqm_and_nash_eqm}
	\end{theorem}
	\begin{proof}
Let $\boldz^*=(\boldzi^*,\boldznoti^*) \in \ne$, the set of Nash equilibria. Then, by definition of Nash equilibrium, $\boldzi^* \in \argmax_{\boldzi \in \constraintseti} U_{i}(\boldzi,\boldznoti^*),\:\forall i \in \agt$. Recall that the unconstrained dynamics~\eqref{eq:uncontrained_agent_dynamics} for each agent $i$ is equal to gradient ascent of its utility $U_i$ with respect to $\boldzi$ assuming $\boldznoti^*$ to be fixed; \emph{i.e.,} $\nabla_{\boldzi}U_i(\boldzi,\boldznoti^*) = -\boldfi(\boldzi,\boldznoti^*)$. From~\cite[Proposition 2.3]{1996_AN_DZ_book}, we can conclude that $\boldzi^* \in \solvarineqfistar,\:\forall i \in \agt$. Now, Lemmas~\ref{lem:reln_betn_var_ineq} and~\ref{lem:eqpt_exist} imply that $\boldz^*\in \eqpt$.
		
Next, observe that $\forall i \in \agt$, the Hessian of $U_{i}(\boldzi,\boldznoti)$ in~\eqref{eq:utility} with respect to $\boldzi$ is a diagonal matrix with its diagonal elements being equal to $-(\wij+\agtsumknoti \aik)$. If Assumption~\ref{asmp:WAR} holds, then it is easy to see that this Hessian matrix is negative semidefinite $\forall \boldzi \in \constraintseti$. Hence, $U_{i}(\cdot,\boldznoti)$ is concave on the set $\constraintseti$, $\forall i \in \agt$.  The rest of the claim can now be proved using similar arguments used in proof of~\cite[Proposition 1.4.2]{2003_FF_PJ_book} for concave functions.
	\end{proof}
	
	We conclude this section by showing that the function $W(\cdot)$ defined earlier in~\eqref{eq:potential_function} is an exact potential function\footnote{See~\cite{2016_QDL_YHC_BHS_potential_game_theory_book} for the definition of potential functions and potential games.} of the game $\Mc{G}$. This is also useful in showing that a Nash equilibrium always exists for the opinion formation game.
	
	\begin{theorem}
		\thmtitle{Opinion formation game is an exact potential game}
		Consider the opinion formation game $\mathcal{G}$ and its set of Nash equilibria $\ne$ defined in~\eqref{eq:ne}. Then $\mathcal{G}$ is an exact potential game with the potential function $W(\cdot)$ defined in~\eqref{eq:potential_function}. Consequently, $\argmax_{\boldz \in \constraintset}W(\boldz) \subseteq \ne \neq \varnothing$.
		\label{thm:exact_potential_game}
	\end{theorem}
	
	\begin{proof}
          For any $i \in \agt$, it can be easily verified from~\eqref{eq:utility} and~\eqref{eq:potential_function} that $\forall \boldx_i,\boldy_i \in \constraintseti$ and $\forall \boldznoti \in \constraintsetnoti$,
		\begin{equation*}
                  W(\boldx_i,\boldznoti)-W(\boldy_i,\boldznoti) = U_i(\boldx_i,\boldznoti)-U_i(\boldy_i,\boldznoti).
                \end{equation*}
Thus, $\Mc{G}$ is an exact potential game with the potential function $W(\boldz)$. The existence of a Nash equilibrium is a consequence of the continuity of $W$ and the compactness of the strategy set $\constraintset$, which implies that the potential function attains a maximum in the set $\constraintset$
		~\cite{2016_QDL_YHC_BHS_potential_game_theory_book}.
	\end{proof}
	
	It is important to note that in an exact potential game with a potential function $W(\boldz)$, $\ne$ is not always equal to $\argmax_{\boldz \in \constraintset}W(\boldz)$. In general $\argmax_{\boldz \in \constraintset}W(\boldz) \subseteq \ne$. The next result, which is a consequence of Theorems~\ref{thm:inv_jacobian_properties}, \ref{thm:reln_betn_eqm_and_nash_eqm}, and~\ref{thm:exact_potential_game} equates the two sets when the Jacobian $\boldJ$ is positive semi-definite. We skip a formal proof.
	
	\begin{corollary}
		\label{cor:equality_of_eqm_sets_and_max_of_pot_fn}
		Suppose $\boldJ$ is positive semi-definite. Then,
		$\ne = \eqpt =\argmax_{\boldz \in \constraintset}W(\boldz) \,$.
		\proofend
	\end{corollary}

	\section{When there are no antagonistic relations}
	\label{sec:analysis_in_absence_of_antagonists}
	Here, we analyze the properties of the unique Nash equilibrium of the opinion formation game in the absence of antagonistic relations; \emph{i.e.,} when Assumption~\ref{asmp:no_antagonistic_relations} holds. Since Assumption~\ref{asmp:no_antagonistic_relations} is a special case of Assumptions~\ref{asmp:WAR_WSI} and~\ref{asmp:WAR}, $\eqpt=\ne$ and the solution of~\eqref{eq:projected_dynamics} converges to the unique equilibrium point $\boldz^* \in \eqpt$ starting from any $\boldz(0) \in \constraintset$. By definition, we know that $\boldz^*=(\boldzi^*,\boldznoti^*) \in \ne$ if and only if $\boldzi^* \in \argmax_{\boldzi \in \constraintseti} U_{i}(\boldzi,\boldznoti^*),\:\forall i \in \agt$. For any agent $i \in \agt$, using~\eqref{eq:utility} let us re-write the function $U_i(\boldzi,\boldznoti^*)$ as
	\begin{equation}
		U_i(\boldzi,\boldznoti^*) = -0.5 \left\|\boldzi-\tildepi(\boldznoti^*) \right\|_{\tildeDi}^{2} - \Delta_{i}(\boldznoti^*),
		\label{eq:utility_of_i_given_strategies_of_others}
	\end{equation}
	where $\tildeDi:=\big(\boldDi + \agtsumknoti\aik \identitym\big) \in \real^{m \times m}_{\geq 0}$ is a positive diagonal matrix, $\tildepi(\boldznoti^*):=\tildeDi^{-1}\big(\boldDi\boldpi + \agtsumknoti\aik \boldzk^*\big) \in \realm_{\geq 0}$ and $2\Delta_{i}(\boldznoti^*):= \boldpi^\top\boldDi\boldpi+\agtsumknoti\aik\boldzk^{*\top}\boldzk^*-\tildepi(\boldznoti^*)^\top \tildeDi \tildepi(\boldznoti^*) \in \real$. Henceforth, we omit the arguments of $\tildepi(\boldznoti^*)$ and $\Delta_i(\boldznoti^*) $ for brevity.
	
Now, consider a particular agent $i \in \agt$. Note that $\tildepi$ and $\Delta_i$ are independent of $i$'s own opinion $\boldzi$ and depends solely on its neighbors' opinions $\boldzk^*,\: k \in \Ni$. Hence, if $\boldznoti^* \in \constraintsetnoti$ denotes the Nash equilibrium opinions of other agents, then from the above discussion, we can see that the Nash equilibrium opinion $\boldzi^*$ of agent $i$ is the unique solution to the quadratic program $\min_{\boldzi \in \constraintseti} 0.5 \left\|\boldzi -\tildepi \right\|_{\tildeDi}^{2}$.  Thus, for each agent $i \in \agt$, its unique Nash equilibrium opinion vector $\boldzi^* \in \constraintseti$ is the weighted projection of the \emph{neighbor influenced preference} vector $\tildepi$ onto the compact, convex set $\constraintseti$. Moreover, for any Nash equilibrium $\boldz^* \in \ne$, we can partition the entire set of agents $\agt$ into two disjoint subsets,
	\begin{subequations}
	\vspace{-1em}
	\begin{align}
		\label{eq:agent_define_exhaust} & \agte(\boldz^*) \ldef \{i\in \agt \mid \boldci^\top \boldzi^* = B_i\}, \\ 
		\label{eq:agent_define_not_exhaust} & \agtde(\boldz^*) \ldef \{i\in \agt \mid \boldci^\top \boldzi^* < B_i\}\,. 
	\end{align}
	\label{eq:agent_define_exhaust_not_exhaust}
	\end{subequations}		
Here, $\agte(\boldz^*)$ is the set of all agents who fully exhaust their resources at Nash equilibrium and $\agtde(\boldz^*) = \agt \setminus \agte(\boldz^*)$ is the set of agents who do not. Viewing each agent's Nash equilibrium opinion $ \boldzi^* $ as the solution to the projection problem provides insight into its properties by partitioning agents into the aforementioned sets $ \agte(\boldz^*) $ and $ \agtde(\boldz^*) $.

We first state the conditions that the Nash equilibrium opinions of agents in these sets must satisfy and then in the remark that follows, discuss the implications.
	\begin{lemma}
		\thmtitle{Necessary conditions on opinions of agents in $\agte$ and $\agtde$}
		Consider the equilibrium points $\eqpt$ in~\eqref{eq:eqpt} of the dynamics $\pds$ in~\eqref{eq:projected_dynamics} and the set of Nash equilibria $\ne$ in~\eqref{eq:ne} of the game $\mathcal{G}$. Suppose Assumption~\ref{asmp:no_antagonistic_relations} holds. Let $\boldz^* \in \ne=\eqpt$ denote the unique Nash equilibrium of $\mathcal{G}$, and let $\agte(\boldz^*)$ and $\agtde(\boldz^*)$ be as defined in~\eqref{eq:agent_define_exhaust_not_exhaust}. 
Further, for any $i \in \agt$ and $s \in \tp$, define $\tilde{w}_{i}^{s}\ldef (w_{i}^{s} + \agtsumknoti \aik)$ and let $(z_{k}^{s})^*$ (resp. $\tilde{p}_{k}^{s}$) denote the $s\tth$ element of $\boldzk^*$ (resp. $\tildepk(\boldznotk^*)$).
Then, for some $\lambda^* \leq 0$,
\begin{subequations}
	\begin{align}
		\label{eq:agent_do_not_exhaust_budget} & \hspace{-1ex} i \in \agtde(\boldz^*) \hspace{-1ex} \implies \hspace{-1ex} \boldzi^* = \tildepi(\boldznoti^*)\,,\\
		\label{eq:agent_exhausts_budget} & \hspace{-1ex} i \in \agte(\boldz^*) \hspace{-1ex} \implies \hspace{-1ex} \frac{\tilde{w}_{i}^{s}	[(z_{i}^{s})^* - \tilde{p}_{i}^{s}]}{c_{i}^{s}} = \lambda^*, \forall s \in \support{\boldz_{i}^*}.
	\end{align}
\end{subequations}

		\label{lem:conditions_on_eqm_opinions_for_agent_to_be_in_ve_vde}
	\end{lemma}
	\begin{proof}
		Under the stated assumptions, if $\boldz^* \in \ne$, then by definition $\tildepi \geq 0,\:\forall i \in \agt$. Also, recall that $\boldzi^* = \argmin_{\boldzi \in \constraintseti} 0.5 \left\|\boldzi -\tildepi \right\|_{\tildeDi}^{2},\:\forall i \in \agt$. This means each $\boldzi^*$ is the unique weighted projection of $\tildepi$ onto closed, convex set $\constraintseti$. It is then easy to see that $i \in \agtde(\boldz^*)$ if and only if $\ci^\top \tildepi < B_{i}$. Since $\tildepi \geq \boldzero$ it means that $\tildepi \in \constraintseti$ and hence $\boldzi^*= \tildepi$. This proves the equality in~\eqref{eq:agent_do_not_exhaust_budget}. Similarly, we can show that $i \in \agte(\boldz^*)$ if and only if $\ci^\top \tildepi \geq B_{i}$. This means that $\boldzi^*$ is the solution to the projection problem only if $\ci^\top \boldzi^* = B_i$. Hence, $\boldzi^* = \argmin_{ \{\boldzi \mid \boldzi \geq \boldzero,\; \ci^\top \boldzi = B_i\} } 0.5 \left\|\boldzi -\tildepi \right\|_{\tildeDi}^{2}$. From the KKT conditions for this (strict) convex program we have,
		\begin{subequations}
			\vspace{-1em}
			\label{eq:DAG_kkt}
			\begin{equation}
				\tildeDi(\boldzi^* - \tildepi) - \lambda^*\ci - \boldsymbol{\mu}^* = \boldzero\;,
				\label{eq:Lagrange}
			\end{equation}
			\begin{equation}
				\boldsymbol{\mu}^{*\top} \boldzi^* = 0
				\label{eq:compl_slackness},
				\boldsymbol{\mu}^* \geq \boldzero,\:\boldzi^* \in \{\boldzi | \boldzi \geq \boldzero, \ci^\top \boldzi = B_i\}.
			\end{equation}
			\label{eq:kkt}
		\end{subequations}
		Then~\eqref{eq:agent_exhausts_budget} follows from~\eqref{eq:kkt}. This completes the proof. 
	\end{proof}	
	
		\begin{remark}
		\thmtitle{Properties of Nash equilibrium in absence of antagonistic relations} 
 Lemma~\ref{lem:conditions_on_eqm_opinions_for_agent_to_be_in_ve_vde} gives us the following insight. From~\eqref{eq:agent_do_not_exhaust_budget} and by the definition of $\tildepi$, we observe that for any agent $i \in \agtde(\boldz^*)$, its equilibrium opinion $(\zij)^*$ on any topic $j$ is equal to a convex combination of $\dij$ and its neighbors' equilibrium opinions $(\zkj)^*$ on that topic. 
Thus, if an agent does not exhaust its budget, it settles for an opinion that is a weighted average of its neighbors steady-state opinions and its internal preference. 

Now, when an agent exhausts its resources, \emph{i.e.} $i \in \agte = \agt \setminus \agtde(\boldz^*)$, Lemma~\ref{lem:conditions_on_eqm_opinions_for_agent_to_be_in_ve_vde} is useful for commenting on the closeness of its equilibrium opinion from its neighbor influenced preference vector. In fact, from~\eqref{eq:agent_exhausts_budget}, we note that for any agent $i \in \agte(\boldz^*)$ such that $\boldzi^* \neq \tildepi$, $\tilde{w}_{i}^{s}/ c_{i}^{s} > \tilde{w}_{i}^{l}/ c_{i}^{l}$ if and only if $ \left|(z_{i}^{s})^* - \tilde{p}_{i}^{s}\right| < \left|(z_{i}^{l})^* - \tilde{p}_{i}^{l}\right|, l,s \in \support{\boldz_i^*}$. This implies that agent $i$'s equilibrium opinion $(z_{i}^{s})^*$ is closest to the neighbor-influenced preference $\tilde{p}_{i}^{s}$ for the topic in $\support{\boldz_i^*}$ with the largest value of $\tilde{w}_{i}^{s}/ c_{i}^{s}$.		
		\remend
	\end{remark}

In light of Lemma~\ref{lem:conditions_on_eqm_opinions_for_agent_to_be_in_ve_vde}, we see the importance of partitioning the agents into $\agte(\boldz^*)$ and $\agtde(\boldz^*)$.	
So, in the remaining portion of the section, 
we give sufficient conditions based solely on model parameters to determine if an agent either fully exhausts or does not exhaust its budget. 
	
\begin{lemma}
  \thmtitle{Tractable sufficient condition for an agent to be in $\agtde$} Consider the equilibria $\eqpt$ in~\eqref{eq:eqpt} of the dynamics $\pds$ in~\eqref{eq:projected_dynamics} and the Nash equilibria $\ne$ in~\eqref{eq:ne} of $\mathcal{G}$. Let Assumption~\ref{asmp:no_antagonistic_relations} hold and let $\boldz^*\in \ne=\eqpt$. Suppose $\exists i \in \agt$ such that $\ci^\top \boldsymbol{\upsilon}_{i} < B_{i}$ where $\boldsymbol{\upsilon}_{i}\ldef[\upsilon_{i}^{1},\ldots,\upsilon_{i}^{m}]$ with $\upsilon_{i}^{s}\ldef (w_{i}^{s} p_i^s + \agtsumknoti \aik \gamma_k B_{k})/\tilde{w}_{i}^{s},\:\forall s \in \tp$ and $\tilde{w}_{i}^{s}$ as defined in Lemma~\ref{lem:conditions_on_eqm_opinions_for_agent_to_be_in_ve_vde} and $\gamma_k \ldef (\min_{l \in \tp} c_k^l )^{-1}>0,\:k \in \agt$. Then, $\ci^\top \boldzi^* < B_i$.
  \label{lem:suff_cond_for_agent_not_exhausting}
\end{lemma}
	\begin{proof}
		Under the stated assumptions and the definitions of $\boldsymbol{\upsilon}_{i}$ and $\tildepi$, it is clear that $\tildepi \leq \boldsymbol{\upsilon}_{i}$. This follows because the equilibrium opinion profiles of neighboring agents, regardless of their exact values, satisfy $\boldzl^* \in \constraintsetl,\:\forall l \in \Ni$. The claim now follows from the proof of Lemma~\ref{lem:conditions_on_eqm_opinions_for_agent_to_be_in_ve_vde}.
	\end{proof}
	
	Note that $\boldsymbol{\upsilon}_i$ depends on the importance weights $\wij$, influence weights $\aik$, and the budgets of neighboring agents $B_k$. So, if there exists an $i \in \agt$ with a large enough budget $B_i$ such that $\mathbf{c}_i^\top\boldsymbol{\upsilon}_i < B_i$, then by Lemma~\ref{lem:suff_cond_for_agent_not_exhausting}, the agent will not exhaust its budget at equilibrium.
	
	Finally, we provide a sufficient condition, again based solely on model parameters, to guarantee that an agent exhausts its budget at the Nash equilibrium. Before proceeding, we state and prove this short result next.

	\begin{lemma}
	Suppose Assumption~\ref{asmp:no_antagonistic_relations} holds. Let $\boldq^* := \boldJ^{-1}\boldD\boldp$. Then $\boldq^*\geq \boldzero$.
		\label{lem:non_negativity_of_unconst_opinions}
	\end{lemma}
	\begin{proof}
		Under Assumption~\ref{asmp:no_antagonistic_relations} it can be shown that $\boldJ$ is an invertible M-matrix. Hence, its inverse $\boldJ^{-1}$ is a non-negative matrix~\cite[Lemma 7]{2017_AP-RT_ARC}. The claim now follows from Assumption~\ref{asmp:param_sign}.
	\end{proof}
	
	Note that since $\boldJ$ is invertible, $\boldq^*$ defined in Lemma~\ref{lem:non_negativity_of_unconst_opinions} is in fact the unique equilibrium point of the unconstrained dynamics~\eqref{eq:uncontrained_dynamics}. Additionally, let $\boldq_{i}^*$ represent the $i\tth$ sub-vector of $\boldq^*$, corresponding to the \emph{unconstrained equilibrium opinion} vector of $i \in \agt$. We now present the final result. 
	
	\begin{lemma}
		\thmtitle{Tractable sufficient condition for an agent to be in $\agte$}
		Consider the equilibrium points $\eqpt$ in~\eqref{eq:eqpt} of the dynamics $\pds$ in~\eqref{eq:projected_dynamics} and the set of Nash equilibria $\ne$ in~\eqref{eq:ne} of the game $\mathcal{G}$. Let Assumption~\ref{asmp:no_antagonistic_relations} hold and let $\boldz^* \in \ne = \eqpt$. Suppose $\exists i \in \agt$ and $j \in \tp$ such that,
		\begin{align}
			(q_{i}^{j})^* &> \left[\frac{\agtsumknoti \aik (q_{k}^{j})^*}{\agtsumknoti \aik} + \frac{B_i}{c_{i}^{j}}\right],
\label{eq:condn_for_agent_in_ve}
		\end{align}
		where for any $i \in \agt$ and $j \in \tp$, $(q_{i}^{j})^*$ denotes the $j\tth$ element of $\boldq_{i}^*$.
		Then, $\ci^\top \boldzi^* = B_i$.
		\label{lem:suff_cond_for_agent_exhausting}
	\end{lemma}
	\begin{proof}
First note that the vector field $\boldf(\boldz)$ defined in~\eqref{eq:uncontrained_dynamics} can be written in terms of $\boldq^*$ as $\boldf(\boldz)= \boldJ(\boldz - \boldq^*)$. Hence $\boldfi(\boldz) = \boldDi(\boldzi - \boldq_{i}^*) - \agtsumknoti\aik \left[\boldq_{i}^* - \boldq_{k}^* + \boldzk - \boldzi \right]$. Also note that from~\eqref{eq:constraint_set}, for any $k \in \agt$ we have $0 \leq \zkj \leq B_k/c_k^j,\:\forall j \in \tp$. From Lemma~\ref{lem:non_negativity_of_unconst_opinions} and the condition given in~\eqref{eq:condn_for_agent_in_ve}, we can show that $f_{i}^{j}(\boldz)<0$. We are now ready to prove the claim via contradiction. Let $\boldz^*\in \eqpt$. From Lemmas~\ref{lem:reln_betn_var_ineq} and~\ref{lem:eqpt_exist}, we have that $\boldzi^*\in \solvarineqfistar$, \emph{i.e.,} $
		\langle \boldfi(\boldzi^{*},\boldznoti^*), \boldx_{i} - \boldzi^{*} \rangle \geq 0, \quad\forall \boldxi \in \constraintseti$. Let the condition in~\eqref{eq:condn_for_agent_in_ve} hold. Suppose $\ci^\top \boldzi^* <B_i$. Choose $s_i^j > (z_i^j)^*$ such that $\ci^\top \boldxi = B_i$, where $\boldxi = [(z_{i}^{1})^{*},\ldots,(z_{i}^{j-1})^{*},s_{i}^{j},(z_{i}^{j+1})^{*},\ldots,(z_{i}^{m})^{*}]^\top \in \constraintseti$, with all values being exactly the same as those in $\boldzi^*$ except the $j\tth$ value, which is equal to $s_i^j$. Clearly, $\boldxi \geq \boldzero$ because such an $s_i^j$ exists since $\ci^\top\boldzi^* < B_i$. Thus, $\boldxi \in \constraintseti$. Now, using this $\boldxi$, evaluating the above inner product gives $\langle \boldfi(\boldzi^{*},\boldznoti^*), \boldx_{i} - \boldzi^{*} \rangle = (s_i^j - (z_i^j)^*) f_i^j(\boldz^*) < 0$. This contradicts the fact that $\boldzi^*\in \solvarineqfistar$.
	\end{proof}
	
	We now interpret the sufficient condition in the above result. Recall from the proof of Lemma~\ref{lem:suff_cond_for_agent_exhausting} that for any agent $i \in \agt$, the upper bound on $\zij \leq B_i/c_i^j,\:\forall j \in \tp$ represents the \emph{maximum possible budget-constrained} opinion that $i$ can hold on topic $j$. The condition in~\eqref{eq:condn_for_agent_in_ve} states that if there exists a topic $j \in \tp$ for which the unconstrained equilibrium opinion $(q_{i}^j)^*$ is sufficiently larger than $B_i/c_i^j$, then the agent will exhaust its resource budget at equilibrium. 

\begin{figure}
	\centering
	\begin{tabular}{cc}
		\includegraphics[trim = 1.3in 3.2in 1.4in 2.8in, clip,width=0.45\linewidth]{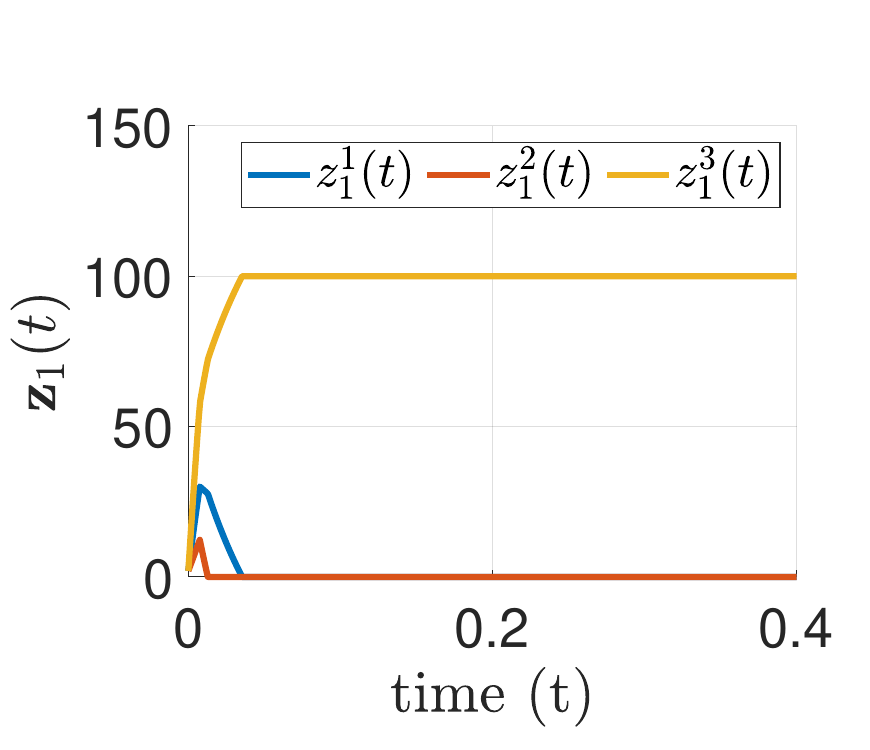} &
		\includegraphics[trim = 1.3in 3.2in 1.4in 2.8in, clip,width=0.45\linewidth]{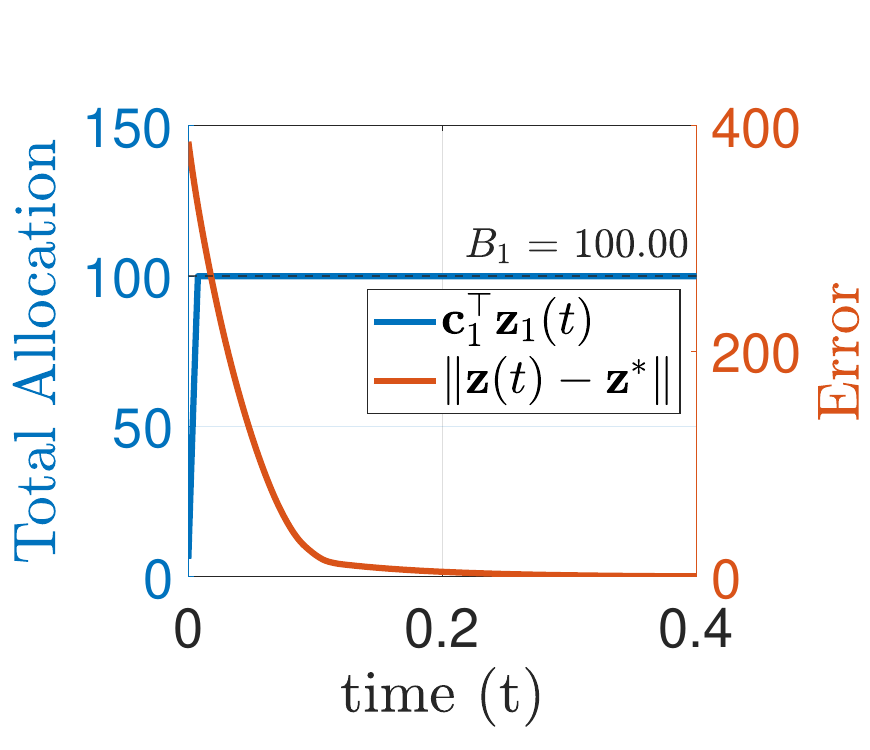}
	\end{tabular}
	\vspace{-2ex}
	\caption{\footnotesize Opinions of agent $1$, which belongs to $\agte(\mathbf{z}^*)$.  
		(Left) Opinions' convergence.  
		(Right) Evolution of total allocation $\mathbf{c}_1^\top \mathbf{z}_1(t)$ and total error $\|\mathbf{z}(t) - \mathbf{z}^*\|$.}
	\label{fig:agent_exhaust}
\end{figure}

	\section{Simulations}
	\label{sec:simulations} In this section, we present simulations to illustrate our analytical results\footnote{
	Values are rounded to two decimals. Simulations were run in \textsc{MATLAB} using \texttt{ode45} with event handling.
	}. 
	We consider a group of four agents forming opinions on three topics according to~\eqref{eq:projected_dynamics}. For all simulations, we let~$\boldci = [1,1,1]^\top,~\forall i \in \agt$ and~$a_{12}\!\approx\!1.24,~a_{13}\!\approx\!4.21,~a_{14}\!\approx\!1.57,~a_{23}\!\approx\!3.90,~a_{24}\!\approx\!5.14,~a_{34}\!\approx\!6.49$, and~$\aii\!=\!0,~\forall i \in \agt$. Thus, Assumption~\ref{asmp:no_antagonistic_relations} holds.

For the first set of simulations in Figure~\ref{fig:agent_exhaust}, the resource budget vector is $\mathbf{r} \approx [100, 183.38, 325.99, 356.12]^\top$. Figure~\ref{fig:agent_exhaust}~(Left) shows that agent $1$'s opinions in all topics converge to the corresponding values of the unique equilibrium point $\boldz^*\in \eqpt = \argmax_{\boldz \in \constraintset}W(\boldz)$. 
Figure~\ref{fig:agent_exhaust}~(Right) shows that the error $\|\boldz(t)-\boldz^*\|$ decreases monotonically to zero. 
This verifies the claims in Theorems~\ref{thm:inv_jacobian_properties},~\ref{thm:convg_to_uniq_eqm} and Corollary~\ref{cor:equality_of_eqm_sets_and_max_of_pot_fn}. 
The unconstrained equilibrium opinion vectors of all agents are:  
$\boldq_{1}^*\!\approx\![216.45, 203.35, 422.27]^\top$,  
$\boldq_{2}^*\!\approx\![146.76, 198.20, 207.26]^\top$,  
$\boldq_{3}^*\!\approx\![170.87, 293.39, 311.20]^\top$,  
$\boldq_{4}^*\!\approx\![160.35, 340.76, 305.09]^\top$.  
From these data, we can verify that $(q_1^3)^*\!\approx\!422.27$ is greater than the expression in square brackets in~\eqref{eq:condn_for_agent_in_ve}, which is approximately equal to 391.45.
Thus, the condition in~\eqref{eq:condn_for_agent_in_ve} is satisfied. Figure~\ref{fig:agent_exhaust}~(Right) also shows that total allocation $\mathbf{c}_1^\top\boldz_1(t)$ of agent $1$ converges to its budget value, $B_{1} =100$ and never exceeds it. Hence, $1 \in \agte(\boldz^*)$. This verifies Lemma~\ref{lem:suff_cond_for_agent_exhausting}. The equilibrium opinions of agent $2$ is $\boldz_{2}^*\!\approx\![0, 103.41, 79.97]^\top$ and its neighbor influenced preference vector is $\tilde{\mathbf{p}}_{2} \approx [25.15, 143.56, 117.04]$. Thus, $2 \in \agte(\boldz^*)$ because $\mathbf{c}_{2}^\top\boldz_{2}^* \approx 183.38$. Also note that, $2,3 \in \support{\boldz_2^*} \subseteq \tp$. The vector containing the values of $\tilde{w}_2^s$ defined in Lemma~\ref{lem:conditions_on_eqm_opinions_for_agent_to_be_in_ve_vde} is $\tilde{\mathbf{w}}_2\approx [14.48, 21.34, 22.89]^\top$. The vector containing the values of $|(z_{2}^s)^* - \tilde{p}_2^s|$ is $\mathbf{d}_2 \approx [25.15, 39.96, 37.25]^\top$. The above data verifies~\eqref{eq:agent_exhausts_budget} and its interpretation below Lemma~\ref{lem:conditions_on_eqm_opinions_for_agent_to_be_in_ve_vde}. 
For the second set of simulations in Figure~\ref{fig:agent_not_exhaust}, $\mathbf{r}\approx [100, 800, 325.99, 356.12]^\top$ is the resource budget vector containing $B_i$'s.  Figure~\ref{fig:agent_not_exhaust}~(Left) shows agent $2$'s opinions' evolution in all three topics. Figure~\ref{fig:agent_not_exhaust}~(Right) shows the evolution of total allocation $\mathbf{c}_{2}^\top\boldz_{2}(t)$. Here, we see that the agent does not exhaust its budget at equilibrium; \emph{i.e.,} $2 \in \agtde(\boldz^*)$. This is because $\mathbf{c}_{2}^\top \boldsymbol{\upsilon}_{2} \approx 653.8 < 800$. This verifies the sufficient condition given in Lemma~\ref{lem:suff_cond_for_agent_not_exhausting}.  The equilibrium opinion vector of second agent is $\boldz_{2}^* \approx [25.15, 143.56, 117.04]$ and its neighbor influenced preference vector is $\tilde{\mathbf{p}}_{2} \approx [25.15, 143.56, 117.04]$. We observe that $\boldz_{2}^*\approx \tilde{\mathbf{p}}_{2}$. This verifies the implication in~\eqref{eq:agent_do_not_exhaust_budget} of Lemma~\ref{lem:conditions_on_eqm_opinions_for_agent_to_be_in_ve_vde}.

\begin{figure}
	\centering
	\begin{tabular}{ll}
		\includegraphics[trim = 1.3in 3.2in 1.4in 2.8in, clip,width=0.45\linewidth]{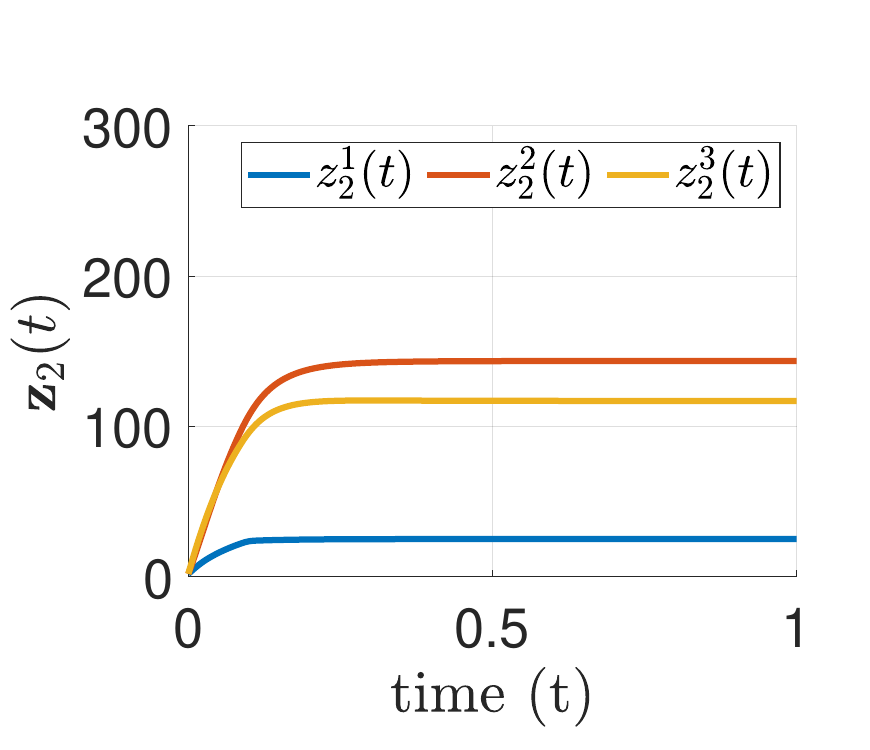} &
		\includegraphics[trim = 1.3in 3.2in 1.4in 2.8in, clip,width=0.45\linewidth]{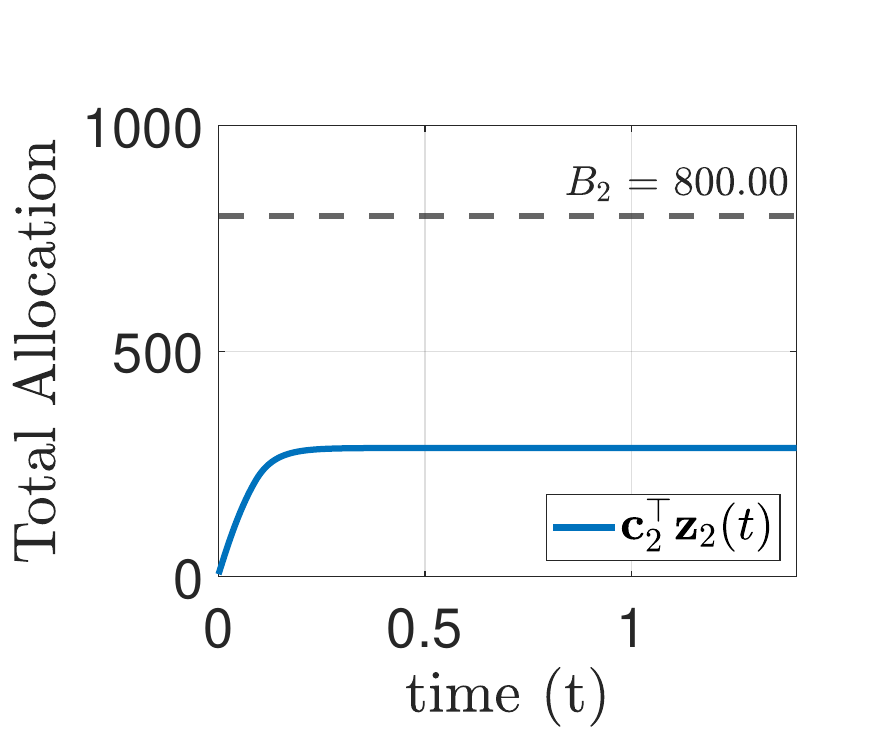}
	\end{tabular}
	\vspace{-2ex}
	\caption{\footnotesize Opinions of agent  $2 \in \agtde(\mathbf{z}^*)$. (Left) Convergence of opinions. (Right) Evolution of total allocation $\mathbf{c}_2^\top \boldz_{2}(t)$.}
	\label{fig:agent_not_exhaust}
\end{figure}

	\section{Conclusions}
	\label{sec:conclusions}
	We introduced a novel opinion dynamics model within the projected dynamical systems framework, where agents allocate limited resources across multiple topics. 
	We analyzed the model’s asymptotic properties, establishing convergence and stability of the equilibrium set.
	We analyzed the underlying opinion formation game, established that it is a potential game and characterized its Nash equilibria.	
	Further, we related the Nash equilibria to the equilibria of the projected dynamics. 
	 Finally, we studied cases where agents exhaust (or do not exhaust) their budget in the absence of antagonistic relations.
	  Potential future work directions include extension to directed signed social networks and exploration of more general utility functions.

	\bibliography{references}	
	\bibliographystyle{IEEEtran}

\end{document}